\documentclass[a4paper,11pt]{article}

\usepackage[a4paper,hmargin={0.9in,0.9in},vmargin=0.9in]{geometry}

\usepackage{amsmath,amssymb,amsthm}

\usepackage[authoryear]{natbib}

\usepackage[colorlinks,citecolor=blue,urlcolor=blue]{hyperref}

\usepackage{graphicx}

\usepackage{rotating}
\usepackage{multicol,multirow,booktabs}
\usepackage{longtable}

\newtheorem{theorem}{Theorem}[section]
\newtheorem{lemma}[theorem]{Lemma}
\newtheorem{proposition}[theorem]{Proposition}
\newtheorem{definition}[theorem]{Definition}

\newcommand{\srho}{\hat{S}^{\rho,\alpha}}
\newcommand{\stau}{\hat{S}^{\tau,\alpha}}

\title{Robust Semiparametric Graphical Models with Skew-Elliptical Distributions}

\author{%
  Gabriele Di Luzio\thanks{Department of Social and Economic Sciences, 
  Sapienza University of Rome, \texttt{gabriele.diluzio@uniroma1.it}}%
  \and
  Giacomo Morelli\thanks{Department of Statistical Sciences, 
  Sapienza University of Rome, \texttt{giacomo.morelli@uniroma1.it}. 
  Corresponding author.}%
}

\date{} 

\begin{document}
\maketitle

\begin{abstract}
We propose semiparametric estimators, called elliptical skew-(S)KEPTIC, for efficiently and robustly estimating non-Gaussian graphical models. Our approach extends the semiparametric elliptical framework to the meta skew-elliptical family, which accommodates skewness.
Theoretically, we show that the elliptical skew-(S)KEPTIC estimators achieve robust convergence rates for both graph recovery and parameter estimation. Through numerical simulations, we illustrate the reliable graph recovery performance of the elliptical skew-(S)KEPTIC estimators. Finally, we apply the new method to the daily log-returns of the stocks in the S\&P 500 index and obtain a sparser graph than with Gaussian copula graphical models.
\end{abstract}





\section{Introduction}
We consider the problem of estimating undirected graphical models. Given a $p$-dimensional random vector $\boldsymbol{X}=(X_1,\hdots,X_p)^{\top}$, we aim to estimate the undirected graph $G:=(V,E)$, where the vertices set $V:=\{1,\hdots,p\}$ contains the number of nodes corresponding to the $p$ variables, and the edge set $E$ indicates the conditional independence relationships between $X_1,\hdots,X_p$. 
Let $X_i$ and $X_j$ be random variables associated with the multivariate distribution $\boldsymbol{X}$ and define $X_{\setminus{\{i,j}\}}$ the set of variables excluding $X_i$ and $X_j$. If $X_i$ and $X_j$ are conditionally independent given $X_{\setminus{\{i,j}\}}$, then there exists no edge between the corresponding nodes $i$ and $j$ in the graph G, denoted as $(i,j) \notin E$.

One reliable approach to estimate the undirected graphical models is the semiparametric {nonparanormal} (NPN) method proposed in \citet{liu2009nonparanormal}, which relaxes the assumption of Gaussian data, considering a set of functions $\{f_j\}_{j=1}^p$ such that: $f(\boldsymbol{X})=(f_1(X_1),\hdots,f_p(X_p))^\top \sim \operatorname{N}(0,\Sigma^0)$, where $\Omega=(\Sigma^{0})^{-1}$. When $\boldsymbol{X}\sim \operatorname{NPN}(\Sigma^0,f)$, no edge connects $X_i$ and $X_j$ if $\Omega_{ij}=0$. 
Uncovering sparsity is subject to the estimation of the unknown correlation matrix. \citet{liu2012high} propose the SKEPTIC (\textit{Spearman/Kendall estimates preempt transformations to infer correlation}) estimator for the {nonparanormal} distribution, a regularized rank-based correlation coefficient estimator based on Kendall's tau and Spearman's rho statistics. 
The estimated correlation matrix is then plugged into shrinkage procedures that reveal the sparsity of the graph. 

The {nonparanormal} distribution exploits a Gaussian-based kernel and belongs to the broad class of {meta-elliptical} distributions described in \citet{han2014scale} and in \citet{liu2012transelliptical}, where a set of functions $\{f_j\}_{j=1}^p$ such that $f(\boldsymbol{X})\sim \mathcal{E}C(0,\Sigma^0,\xi)$ is considered. In particular,  \citet{liu2012transelliptical} show that the graph $G$ is encoded in the latent partial correlation matrix $\Xi:=-\left[\text{diag}(\Omega)\right]^{-1/2}\Omega\left[\text{diag}(\Omega)\right]^{-1/2}$, thus no edge connects $X_i$ and $X_j$ if $\Xi_{ij}=0$, which implies $\Omega_{ij}=0$. A major drawback of the {meta-elliptical} class is that, while it includes skewed distributions, the dependence structure of its copula does not contain the skewness parameters, preventing them from playing a role in identifying the conditional independence structure.

We propose rank-based correlation matrix estimators, which we refer to as elliptical skew-(S)KEPTIC estimators. These estimators refer to the second moment of the {meta skew-elliptical} family of distributions, that we introduce in this paper considering a set of functions $\{f_j\}_{j=1}^p$ such that $f(\boldsymbol{X})\sim \operatorname{SU}\mathcal{E}(0,\Sigma^0,\Lambda,h^{(p+q)},\tau,\Gamma)$, where the unified skew-elliptical ($\operatorname{SU}\mathcal{E}$) distribution is described in \citet{arellano2010skew-e}. Extending the SKEPTIC estimator to the meta skew-elliptical family of distributions, our methodology allows for 
constructing graphical models by using a copula approach to estimate the second moment of the distribution taking into account the univariate directional skewness of the data.
To do so, we consider the unified skew-elliptical distribution as a combination of the closed skew-normal distribution introduced in \citet{dominguez2003multivariate}, which belongs to the class of $\operatorname{SU\mathcal{E}}$  distributions. 
The closed skew-normal distribution admits a stochastic representation involving both a normal and a half-normal distribution, for which the SKEPTIC estimator of the correlation matrix is defined. The elliptical skew-(S)KEPTIC extends this estimator by incorporating a correction term that accounts for the directional skewness of the closed skew-normal distribution. Remarkably, we show that, given a bivariate relationship, concordant skewness directions strengthen the dependence between the variables, whereas discordant skewness directions weaken it.

The estimated correlation matrix is then plugged into existing shrinkage procedures, see, e.g., (\citet{meinshausen2006high}; \citet{banerjee2008model}; \citet{friedman2008sparse}; \citet{yuan2010high}; \citet{cai2011constrained}). 
The elliptical skew-(S)KEPTIC estimators that generalize the SKEPTIC estimator are based on Kendall's tau and Spearman's rho statistics. 
When the unified skew-elliptical distribution reduces to the closed skew-normal distribution, we consider both the statistics for estimating the unknown correlation matrix and we define the elliptical skew-SKEPTIC estimator. Otherwise, we only consider Kendall's tau statistic and we define the elliptical skew-KEPTIC estimator, since Spearman's rho statistic is not invariant for the class of elliptical distributions (\citealp{hult2002multivariate}).


We prove that the elliptical skew-(S)KEPTIC estimators have the same rate of convergence of the SKEPTIC estimator of \cite{liu2012high} plus a error rate from the estimate of the skewness parameters. The rates of convergence are studied in $\|\cdot \|_{\max}$ norm.

We provide a numerical analysis to support the theory. Through a backward approach, we define the true sparsity pattern given by $\Omega$ and show the graph recovery using the elliptical skew-(S)KEPTIC estimators, which are compared to the SKEPTIC estimator. The elliptical skew-(S)KEPTIC estimators results in similar  graph recovery as the SKEPTIC, which is proven to be optimal in \cite{liu2012high}, showing statistical efficiency of the proposed method.

Prior studies by \citet{zareifard2016skew}, \citet{nghiem2022estimation}, and \citet{sheng2023skewed} have explored graphical models while incorporating skewness into the sparsity pattern estimation. 
\citet{zareifard2016skew} exploit a multivariate closed skew-normal distribution to define the skew Gaussian graphical models, which are estimated through a Bayesian approach. \citet{nghiem2022estimation} propose a novel nodewise regression approach to estimate the graph on data generated from a generalized multivariate skew-normal distribution.
Instead, \citet{sheng2023skewed} include the shape parameter in the determination of the precision matrix using the skew-normal distribution of \citet{azzalini1985class} and developing an algorithm that penalizes the likelihood of this distribution to estimate $\Omega$ and, consequently, the graph. 
However, these studies do not address the direct estimation of the correlation matrix. By contrast, we propose an estimator that explicitly incorporates the skewness parameters in the dependence structure through a copula-based approach.

The paper is organized as follows. In Section \ref{sec2}, we describe the background, including a review of the SKEPTIC estimator
and the semiparametric elliptical distribution. 
In Section \ref{sec3}, we introduce the meta skew-elliptical distribution, which forms the basis to derive the elliptical skew-(S)KEPTIC estimators whereas in Section \ref{shrinkage_procedures}, we provide some shrinkage procedures to estimate the precision matrix $\Omega$. Theoretical properties of the proposed estimator are discussed in Section \ref{sec4}. Section \ref{sec5}, presents simulation analysis comparing the elliptical skew-(S)KEPTIC estimators with the SKEPTIC estimator on synthetic data, and reports an empirical application to S\&P 500 stock returns.  Section \ref{sec7} concludes the paper.

\section{Background}\label{sec2}
In this section, we briefly recall the literature preceding our article: the SKEPTIC estimator of \citet{liu2012high}, the semiparametric elliptical distributions in \cite{liu2012transelliptical} and \citet{han2014scale}. The section concludes with the introduction of the unified skew-elliptical family of distributions in \citet{arellano2010skew-e}, and its link with the undirected graphical models.
\subsection{\textit{Notation}}
Here is the notation used throughout the paper. Let $A=\left[A_{ij}\right]\in \mathbb{R}^{p\times p}$ and $\boldsymbol{a}=(a_1,\hdots,a_p)\in \mathbb{R}^p$. For $1\leq q <\infty$, we define $\|\boldsymbol{a}\|_{q}=\left(\sum_{i=1}^p|\boldsymbol{a}|^q\right)^{1/q}$, when $q=\infty$: $\|\boldsymbol{a}\|_{\infty}=\max_{1\leq i\leq p}|\boldsymbol{a}_{i}|$. Regarding the $A$ matrix, when $q=1$: $\|A\|_{1} = \max_{1\leq j\leq q}\sum_{i=1}^p |A_{ij}|$, and when $q=\infty$, $\|A\|_{\infty}=\max_{1\leq i \leq q}\sum_{j=1}^p |A_{ij}|$. We denote $\|A\|_{\max}=\max_{i,j}|A_{ij}|$. A vector without its $j$-th element is defined: $\boldsymbol{a}_{\backslash j}=(\boldsymbol{a}_{1},\hdots, \boldsymbol{a}_{j-1},\boldsymbol{a}_{j+1},\hdots,\boldsymbol{a}_{p})^\top \in \mathbb{R}^{p-1}$, and we denote $A_{\backslash i,\backslash j} \in \mathbb{R}^{p-1 \times p-1}$ the matrix without the $i$-$th$ row and the $j$-th column. We define: $\lambda_{\min}(A)$ and $\lambda_{\max}(A)$, respectively, the minimum and the maximum eigenvalue of $A$. $\Omega$ denotes a sparse matrix, and we define the degree of $\Omega$: $\operatorname{deg}(\Omega)=\max_{1\leq i\leq p} \sum_{j=1}^p\operatorname{I}(|\Omega_{ij}|\neq 0)$, where $i \neq j$.
\subsection{The Nonparanormal\label{sec: npn} SKEPTIC}

A random vector $\boldsymbol{X}=(X_1,...,X_p)^{\top}$ has a {nonparanormal} distribution if there exists functions $\{f_j\}_{j=1}^p$ such that $\boldsymbol{Z}\equiv f(\boldsymbol{X})\sim \operatorname{N}(0,\Sigma^{0})$, where $f(\boldsymbol{X})=(f_1(X_1),...,f_p(X_p))^{\top}$, we then write $\boldsymbol{X} \sim \operatorname{NPN}(\Sigma^0,f)$. Further, if the functions $\{f_j\}_{j=1}^p$ are monotone and differentiable, the joint distribution of $\boldsymbol{X}$ is a Gaussian copula. The focal point is to find a function, $f(\boldsymbol{X})$, that transforms our data into normally distributed data.

Given \( n \) observations of the \( j \)-th variable, \( x_{1j}, \dots, x_{nj} \), we define the empirical cumulative distribution function (ECDF) $\widehat{F}_j(w) = \frac{1}{n} \sum_{i=1}^{n} \mathbb{I}(x_{ij} \leq w),$
where \( \mathbb{I}(\cdot) \) is the indicator function. The transformation function \( \widehat{f}_j(x_{ij}) \) is then defined as  $\widehat{f}_j(x_{ij}) = \Phi^{-1}(\widehat{F}_j(x_{ij}))$,
where \( \Phi^{-1}(\cdot) \) denotes the quantile normal distribution function. In practice, the ECDF at each observation \( x_{ij} \) can be approximated using its normalize rank:  $\widehat{F}_j(x_{ij}) \approx \widehat{r}_{ij} = \frac{r_{ij}}{n+1}$,
where \( r_{ij} \) is the rank of \( x_{ij} \) among \( x_{1j}, \dots, x_{nj} \). Thus, the transformation results in:  
\begin{equation}\label{eq: normal_score}
    \widehat{f}_j(x_{ij}) = \Phi^{-1} \left( \widehat{r}_{ij}\right).
\end{equation}  
This approach ensures that the transformed data approximate a standard normal distribution while preserving the rank structure of the original data. We consider the following statistics:
\begin{align}
    &\text{(Spearman's rho)}\ \widehat{\rho}_{ij}=\frac{\sum_{k=1}^n(r^k_j-\bar{r}_j)(r^k_i-\bar{r}_i)}{\sqrt{\sum_{k=1}^n(r^k_j-\bar{r}_j)^2\sum_{k=1}^n(r^k_i-\bar{r}_i)^2}},\\
    &\text{(Kendall's tau)}\ \widehat{\tau}_{ij}=\frac{2}{n(n-1)}\sum_{1\leq k<k'\leq n}\text{sign}((x^k_j-x^{k'}_j)(x^k_i-x^{k'}_i)).
\end{align}
Both $\widehat{\rho}_{ij}$ and $\widehat{\tau}_{ij}$ are non-parametric correlations between the empirical realizations of random variables $X_j$ and $X_i$. Let \(\widetilde{X}_i\) and \(\widetilde{X}_j\) be independent copies of the random variables \(X_i\) and \(X_j\),  \(F_i\) and \(F_j\) denote the cumulative distribution functions (CDFs) of \(X_i\) and \(X_j\), the population versions of Spearman’s rho and Kendall’s tau are defined as:
$\rho_{ij} := \operatorname{Corr}(F_i(X_i), F_j(X_i))$ and
$\tau_{ij} := \operatorname{Corr}(\operatorname{sign}(X_i - \widetilde{X}_i), \operatorname{sign}(X_j - \widetilde{X}_j))$,
respectively.
For {nonparanormal} distributions, there is a connection between Spearman's rho and Kendall's tau to estimate the underlying Pearson's correlation coefficient $\Sigma^0_{ij}$, that we  recall.
\begin{lemma}\label{Lemma: SKEPTIC}
(\citet{kruskal1958ordinal}, \citet{liu2012high}):  Suppose that there exists a set of $f=\left\{f_1, \ldots, f_p\right\}$ monotonic functions such that $\boldsymbol{X}\sim \operatorname{NPN}(\Sigma^0,f)$, then we have $\Sigma_{ij}^0=2\operatorname{sin}\left(\frac{\pi}{6}\rho_{ij}\right)=\operatorname{sin}\left(\frac{\pi}{2}\tau_{ij}\right)$.
\end{lemma}
From these results two estimators of the unknown correlation matrix $(\Sigma^0_{ij})$  are $\widehat{S}^{\rho}=\left[\widehat{S}^{\rho}_{ij}\right]$ and $\widehat{S}^{\tau}=\left[\widehat{S}^{\tau}_{ij}\right]$, where: $\widehat{S}^{\rho}_{ij}=2\operatorname{sin}\left(\frac{\pi}{6}\widehat{\rho}_{ij}\right)\cdot \operatorname{I}(i\neq j)+\operatorname{I}(i=j)$ and $\widehat{S}^{\tau}_{ij}=\operatorname{sin}\left(\frac{\pi}{2}\widehat{\tau}_{ij}\right)\cdot \operatorname{I}(i\neq j)+\operatorname{I}(i=j)$.
Spearman's rho and Kendall's tau statistics are invariant under strictly increasing transformations of the underlying random variable. 
\subsection{The {meta-elliptical} distribution}
The {meta-elliptical} distribution introduced in \citet{han2014scale} relaxes the assumption of elliptically distributed data, in a similar way as the {nonparanormal} of \citet{liu2009nonparanormal} does for Gaussian distributions. The {meta-elliptical} family is a strict extension of the {nonparanormal} distribution. 

\begin{definition}
    (\citet{han2014scale}, \citet{fang2018symmetric}) A random vector $\boldsymbol{X}$ has an elliptical distribution, $\boldsymbol{X}\sim \mathcal{E}C_{p}(\mu,\Sigma,\xi)$, if $\boldsymbol{X}$ has the following stochastic representation: $\boldsymbol{X}\overset{d}{=}\mu+\xi A U$. Where: $\mu \in \mathbb{R}^p$, $q:=rank(A)$, $A$ is  a $p \times q$ matrix such that $\Sigma=AA^\top$, $\xi\geq 0$ is a random variable independent of $U$, $U\in \mathbb{S}^{q-1}$, which is uniformly distributed of the unit sphere in $\mathbb{R}^q$.  
\end{definition}

If $\xi$ has a density, then $\boldsymbol{X}$ has a density of the form: 
\begin{equation}\label{ell_disr}
    f_p\left(x ; \mu, \Sigma, h^{(p)}\right)=|\Sigma|^{-1 / 2} h^{(p)}((x-\mu)^\top\Sigma^{-1}(x-\mu)),
\end{equation}
where $h^{(p)}$ is a scale function uniquely determined by the distribution of $\xi$, and it is called density generator.

Let $\mathcal{R}_p : = \{\Sigma \in \mathbb{R}^{p \times p}:\Sigma^{\top}=\Sigma, \text{diag}\left(\Sigma\right)=1, \Sigma \succeq 0\}$, in analogy to the {nonparanormal} case, the {meta-elliptical} distribution is defined as follows.
\begin{definition}(\citet{han2014scale})
    A continuous random vector $\boldsymbol{X}=(X_1,\hdots,X_p)^{\top}$ has a {meta-elliptical} distribution, denoted by $\boldsymbol{X}\sim M\mathcal{E}_p(\Sigma^0, \xi;f)$, if there exist a set of univariate monotone functions $f_1,\hdots,f_p$ such that $(f_1(X_1),\hdots,f_p(X_p))^{\top}\sim \mathcal{EC}_p(0,\Sigma^0,\xi)$, where $\Sigma^0 \in \mathcal{R}_p$ is the {latent  correlation matrix}. 
\end{definition}
For the distributions that belong to the {meta-elliptical} family, Kendall's tau statistic in Lemma \ref{Lemma: SKEPTIC} can be exploited to estimate the unknown correlation matrix.
Nevertheless, the meaning of a missing edge is different for the {nonparanormal} case. Indeed, when $\boldsymbol{X}\sim \operatorname{NPN}(\Sigma^{0},f)$ a missing edge implies conditional independence. On the other hand, when $\boldsymbol{X} \sim M\mathcal{E}(\Sigma^{0},\xi,f)$, if $\Omega_{ij}=0$, then $X_i$ and $X_j$ are conditionally uncorrelated (\citealp{rossell2021dependence}).

 The {meta-elliptical} family contains many useful distributions, among which asymmetric distributions. However, the estimate of the correlation matrix does not directly take into account the skewness parameters. We extend the {meta-elliptical} family giving a direct contribution of the skewness parameters for the conditional independence relations, and thus considering the SU$\mathcal{E}$ family of distributions.

\subsection{Unified Skew-Elliptical (SU$\mathcal{E}$) distributions}
We recall the family of $\operatorname{SU\mathcal{E}}$ distributions described in \citet{arellano2010skew-e}. If a random vector $\boldsymbol{X}$ has a $\operatorname{SU\mathcal{E}}$ distribution, we write $\boldsymbol{X}\sim \operatorname{SU\mathcal{E}}_{p,p}(\mu,\Sigma,\Lambda,h^{(p)},0,\Gamma)$. We focus on two specifications of the family, which are: the closed skew-normal ({CSN}) distribution of \citet{dominguez2003multivariate}\footnote{We build on the closed skew-normal (CSN) distribution, which represents a reparametrization of the unified skew-normal (SUN) distribution employed by \citet{arellano2010skew-e}.} and the unified skew-$t$ (SUT) distribution (\citet{arellano2010skew-t}). Table \ref{tab: SUE} shows how the $\operatorname{SU\mathcal{E}}$ distribution reduces into {CSN}  and {SUT} distributions. Formally, we define the $\operatorname{SU\mathcal{E}}$ distribtuion.

\begin{definition}(\citealp{arellano2010skew-e})\label{def_sue}
A continuous p-dimensional random vector $\boldsymbol{X}$ has a multivariate unified skew-elliptical $\operatorname{(SU\mathcal{E})}$ distribution, denoted by $\boldsymbol{X}\sim SU\mathcal{E}_{p,p}(\mu,\Sigma,\Lambda,h^{(p)},0,\Gamma)$, if its density function at $x \in \mathbb{R}^p$ is:
\begin{equation}
   f_{X}(x) = 2^p f_p\left(x ; \mu, \Sigma, h^{(p)}\right) F_p\left(\Lambda \sigma^{-1}(x-\mu) ; \Gamma, h_{Q(z)}^{(p)}\right) ,
\end{equation}
 where $f_p\left(x ; \mu, \Sigma, h^{(p)}\right)$ is the elliptical distribution defined in \eqref{ell_disr}, and  $ F_p\left(\Lambda \sigma^{-1}(x-\mu) ; \Gamma, h^{(p)}_{Q(x)}\right)$ the $p$-dimensional centered elliptical cumulative distribution function with $p \times p$ dispersion matrix $\Gamma$, $\Lambda$ is a $p \times p$ real matrix controlling shape. In conclusion, $Q(\mathbf{x})$ is a quadratic form defined as $Q(\mathbf{x})=(\mathbf{x}-\mu)^{\top} \boldsymbol{\Sigma}^{-1}(\mathbf{x}-\mu) \in \mathbb{R}^{+}$, whereas $h_{Q(\mathbf{x})}^{(p)}(u)=h^{(2p)}[u+Q(\mathbf{x})] / g^{(p)}[Q(\mathbf{x})]$ denotes the elliptical conditional density generator.
\end{definition}
\begin{table}[h]
\centering
\begin{tabular}{ c|l }
  $\operatorname{SU\mathcal{E}}$ & \(h^{(p)}\)  \\
  \hline
  {CSN} & \(h^{(p)} = \phi^{(p)}\)$(u)$ \quad \text{where:} \(\phi^{(p)}(u) = \frac{1}{(2\pi)^{p/2}}\exp\left(-\frac{u^2}{2}\right)\)  \\
  
  {SUT} & \(h^{(p)} = t^{(p)}(u,\nu)\) \text{   where:} \(t^{(p)}(u,\nu) = \frac{\Gamma\left(\frac{\nu+p}{2}\right)}{\Gamma\left(\frac{\nu}{2}\right)\left(\pi a\right)^{p/2}} \left(1 + \frac{u}{\nu}\right)^{-\frac{\nu+p}{2}}\) 
\end{tabular}
\caption{Subclasses of skew-elliptical distribution based on the density generator \(h^{(p)}\).}\label{tab: SUE}
\end{table} 
The relation between the {CSN} and the $\operatorname{SU\mathcal{E}}$ is defined through the following stochastic representation. 

\begin{proposition}(\citealt{arellano2010skew-e})\label{sue-sun}
Given $\boldsymbol{X}\sim \operatorname{SU\mathcal{E
}}_{p,p}(\mu,\Sigma,\Lambda,h^{(p)},0,\Gamma)$, we can write $\boldsymbol{X}$ as a mixture of a {CSN} distribution: 
$\boldsymbol{X}=\mu+\sigma  V_0^{-1/2}\boldsymbol{Z}_0$, where $V_0$ is the cumulative distribution function of a uni-dimensional random variable, called mixture variable, such that: $V_0(0)=0$, and $V_0$ is independent of $\boldsymbol{Z}_0\sim \operatorname{CSN}_{p,p}(0,\Sigma^0,\Lambda,0,\Gamma)$.
\end{proposition}
\begin{lemma}
(\citealp{arellano2010skew-e})\label{sue-sut}  If $\text{ $V_0$}\sim Gamma(\nu/2,\nu/2)$, then the unified skew-elliptical distribution $\boldsymbol{X}\sim \operatorname{SU}\mathcal{E}_{p,p}(\mu,\Sigma,\Lambda,h^{(p)},0,\Gamma)$ reduces to a unified skew-$t$ $\operatorname{(SUT)}$ distribution: $\boldsymbol{X}\sim \operatorname{SUT}_{p,p}(\mu,\Sigma,\Lambda,\nu,0,\Gamma)$. 
\end{lemma}
We construct an undirected graph on the closed skew-normal distribution, following \citet{zareifard2016skew}.

\begin{definition}\label{def:SGDG}
    Consider an undirected graph $G=(V, E)$, a random vector $\boldsymbol{X}=(X_1,\hdots, X_p)^{\top}$ is called a skew graphical model concerning the graph G with mean $\mu$, the precision matrix $\Omega$ and the skewness parameter $\alpha \in \mathbb{R}^p$ if its density is $\operatorname{CSN}_{p,p}(0, \Omega^{-1}, D_{\alpha}L,0, D_{\kappa})$ or equivalently
     \begin{equation}\label{eq:jointcsn}
        f_{\boldsymbol{X}}(x)=2^p\phi_{p}\left(x-\mu;0,\Omega^{-1}\right)\Phi_{p}\left(D_{\alpha}L(x-\mu);0,D^{-1}_{\kappa}\right),
    \end{equation}
where $L$ and $D_{\kappa}$ are the matrix resulting from the modified Cholesky of $\Omega=LD_{\kappa}L^{\top}$, and zero in $L$ matrix determines the conditional independence relations between the $X_1,\hdots,X_p$ variables. Further, when the matrix $D_{\alpha}=0$, the random vector $\boldsymbol{X}\sim \operatorname{CSN}_{p,p}(0, \Omega^{-1}, D_{\alpha}L,0, D_{\kappa})$ reduces to $\boldsymbol{X}\sim \operatorname{N}_{p}(0,\Omega^{-1})$. 
\end{definition}

The graphical model for the unified skew-elliptical distribution can be constructed by utilizing the sparsity structure of the closed skew-normal distribution. 
We employ a semiparametric approach, utilizing the class of unified skew-elliptical distributions, to extend the SKEPTIC estimator proposed by \citet{liu2012high}. This approach enables the development of robust estimators for graphical models that quantify the skewness present in the data distributions.   
\section{The elliptical skew-(S)KEPTIC estimators}\label{sec3}
 
\begin{definition}(Meta skew-elliptical distribution)
    A continuous random vector $\boldsymbol{X}=(X_1,\hdots,X_p)^{\top}$ has meta skew-elliptical distribution, denoted by $\boldsymbol{X}\sim MS\mathcal{E}_{p,p}(\Sigma^0,\Lambda,h^{(p)},0,\Gamma; f)$, if there exists a set of monotone univariate functions $f_1,\hdots,f_p$ such that
    \begin{equation}
        (f_1(X_1),\hdots,f_p(X_p))^{\top}\sim \operatorname{SU\mathcal{E}}_{p,p}(0,\Sigma^0,\Lambda,h^{(p)},0,\Gamma),
    \end{equation}
    where $\Lambda$ is a matrix containing the skewness terms. 
\end{definition}
\begin{proposition}
    When $\Lambda=0$, then the {meta skew-elliptical} distribution reduces to the {meta-elliptical} distribution. 
\end{proposition} 
\begin{proposition}
    When $\Lambda=0$ and $h^{(p)}=\phi^{(p)}(u)$, then the {meta skew-elliptical} distribution reduces to the nonparanornmal distribution. 
\end{proposition} 
We employ the stochastic representation of the unified skew-elliptical distribution, characterized by a closed skew-normal distribution $\boldsymbol{X}\sim \operatorname{CSN}_{p,p}(0, \Omega^{-1}, D_{\alpha}L,0, D_{\kappa})$ (Proposition \ref{sue-sun}), to introduce the elliptical skew-(S)KEPTIC estimators for constructing skew-elliptical graphical models. The closed skew-normal distribution in \eqref{eq:jointcsn} can be expressed in the following stochastic form:
\begin{equation}\label{eq: stoc_rep}
\boldsymbol{X} \stackrel{d}{=} L^{-1} D_\kappa^{-\frac{1}{2}} D_\alpha\left(I+D_\alpha^2\right)^{-\frac{1}{2}} \boldsymbol{U}+L^{-1} D_\kappa^{-\frac{1}{2}}\left(I+D_\alpha^2\right)^{-\frac{1}{2}} \boldsymbol{V},
\end{equation}
where: $\boldsymbol{U}\sim HN_p(0,I_p)$ and $\boldsymbol{V}\sim N_p(0,I_p)$, where $HN_p(0,I_p)$ is a $p$-dimensional half-normal distribution, and $N_p(0,I_p)$ is a $p$-dimensional normal distribution, the distributions of $\boldsymbol{U}$ and $\boldsymbol{V}$ are independent. We introduce a correlation matrix estimator that explicitly accounts for the cross effects of univariate skewness parameters on the linear dependence among variables. The robustness of the proposed estimators follows from the direct estimation of the skewness parameters without the explicit computation of the marginal distributions, in the spirit of the SKEPTIC estimator of \cite{liu2012high}.

\begin{theorem}[elliptical skew-KEPTIC estimator]\label{main th}
Assume that $\boldsymbol{X}$ follows a meta skew-elliptical distribution: $\boldsymbol{X}\sim MS\mathcal{E}_{p,p}(\Omega^{-1},D_\alpha L,h^{(p)},0,D_\kappa; f)$, such that  $\boldsymbol{Y}=f(\boldsymbol{X})$ denotes the unified skew-elliptical distribution $\boldsymbol{Y}\sim \operatorname{SU\mathcal{E}}_{p,p}(0,\Omega^{-1},D_\alpha L,h^{(p)},0,D_\kappa)$. The elliptical skew-KEPTIC estimator is:
\begin{equation}\notag
   {S}^{\tau,\alpha }_{ij}=\begin{cases}\sin\Big(\frac{\pi}{2}\tau_{ij}\Big)\frac{1+\alpha_i \alpha_j}{\sqrt{(1+\alpha_i^2)(1+\alpha_j^2)}}
    &i\neq j\\ 1&i=j\end{cases} .
\end{equation}
\end{theorem}

 As discussed in \citet{hult2002multivariate}, Spearman's rho is not invariant in the class of elliptical distributions, and its relation with the Pearson correlation coefficient in Lemma \ref{Lemma: SKEPTIC} holds if and only if the distribution is {nonparanormal}.
The proof of the Theorem is in \ref{appendix1}.

\begin{lemma}\label{Lemma_sp_rho}(elliptical skew-SKEPTIC estimator)
Assume that $\boldsymbol{X}$ follows a meta skew-elliptical distribution: $\boldsymbol{X}\sim MS\mathcal{E}_{p,p}(\Omega^{-1},D_\alpha L,h^{(p)},0,D_\kappa; f)$, such that  $\boldsymbol{Y}=f(\boldsymbol{X})$ denotes the unified skew-elliptical distribution. When $\boldsymbol{Y}\sim \operatorname{SU\mathcal{E}}_{p,p}(0,\Omega^{-1},D_\alpha L,h^{(p)},0,D_\kappa)$ reduces to the {CSN} distribution in \eqref{eq:jointcsn}, the elliptical skew-SKEPTIC estimator is:
\begin{equation}\notag
{S}^{\rho,\alpha}_{ij} = \begin{cases}2\sin\Big(\frac{\pi}{6}\rho_{ij}\Big)\frac{1+\alpha_i \alpha_j}{\sqrt{(1+\alpha_i^2)(1+\alpha_j^2)}}
    &i\neq j\\ 1&i=j\end{cases} ,  \quad {S}^{\tau,\alpha }_{ij}=\begin{cases}\sin\Big(\frac{\pi}{2}\tau_{ij}\Big)\frac{1+\alpha_i \alpha_j}{\sqrt{(1+\alpha_i^2)(1+\alpha_j^2)}}
    &i\neq j\\ 1&i=j\end{cases} .
\end{equation}
\end{lemma}
We generally refer to the estimators in Theorem \ref{main th} and Lemma \ref{Lemma_sp_rho} as elliptical skew-(S)KEPTIC estimators. By construction, the estimated dependence measures are amplified when the skewness components exhibit concordance, whereas they are attenuated when the skewness components are discordant.
To obtain the correlation matrix using the elliptical skew-(S)KEPTIC estimators, we need to estimate the parameters of the
 distributions, which are encoded in the diagonal elements of $D_\alpha$, for closed skew-normal and unified skew-$t$ distributions. 
\section{Methods for the Estimation of the Precision Matrix and the Undirected Graph}\label{shrinkage_procedures}



\subsection{The Dantzig selector}
The Dantzig selector estimator (\citealp{yuan2010high}) exploits the relation between the coefficients $\theta$ stemming from a multivariate linear regression and the elements of the precision matrix. A sparse precision matrix implies sparse regression coefficients estimation that relies on an algorithm that depends on the tuning parameter $\delta$. The Dantzig procedure involves two steps:
\begin{itemize}
    \item Estimation: for $j=1,\hdots,p$, calculate
    \begin{align}\label{theta}
        \widehat{\theta}^{j}&=\text{arg}\min_{\theta\in \mathbb{R}^{p-1}}\|\theta\|_{1} \quad \text{subject to  } \quad \|\widehat{S}_{\backslash j,j}-\widehat{S}_{\backslash j,\backslash j}\theta\|_{\infty}\leq \delta,\\\notag
        \widehat{\Omega}_{jj}&=\left[1-2(\widehat{\theta}^{j})\widehat{S}_{\backslash j,j}+(\widehat{\theta}^{j})^{\top}\widehat{S}_{\backslash j,\backslash j}\widehat{\theta}^{j}\right]\quad \text{and } \quad \widehat{\Omega}_{\backslash j.j}=-\widehat{\Omega}_{\backslash jj}\widehat{\omega}^{j}.
    \end{align}
    \item Symmetrization: 
    \begin{align}
        \widehat{\Omega}^{gDS}=\text{arg}\min_{\Omega=\Omega^T}\|\Omega-\widehat{\Omega}\|_1.
    \end{align}
\end{itemize}
\subsection{The Glasso Algorithm}
 Glasso algorithm is developed by \citet{friedman2008sparse} and it consists in adding a $L_1$-norm penalty term in the maximum likelihood estimation of $\Omega$ as follows: 
\begin{equation}\label{eq: glasso}
    \Omega^{glasso}=arg\max_{\Omega\succ 0}\left\{\text{log}\ \text{det}\Omega-\text{tr}(S\Omega)-\lambda_{G} \|\Omega\|_{1}\right\}.
\end{equation}
The penalty term, $\lambda_G$, is the parameter that controls the sparsity of the estimated matrix, i.e., the number of zeros in the precision matrix. \citet{banerjee2008model} shows that the shrinking procedure based on the tuning parameter $\lambda_G$ is directly proportional to the input covariance matrix. This means that greater values of $\lambda_G$ imply a greater presence of zeros in the estimated $\Omega$, therefore, a minor number of interconnections are unveiled.

\subsection{The CLIME Algorithm}
CLIME stands for constrained $\ell_{1}$-minimization for inverse matrix estimation. This algorithm, proposed by \citet{cai2011constrained}, is very attractive for high-dimensional data since the estimator is obtained by combining vector solutions by solving a linear program. Let $\Omega$ denote a $p\times p$ precision matrix and $\{\widehat{\Omega}_1\}$ be the solution set of the optimization problem:
\begin{align}\label{eq:clime}
    \text{min}\|\Omega\|_1 \quad \text{subject to:}\\\notag
    \|\Sigma \Omega-I\|_{\infty}\leq \lambda_C,
\end{align}
where $\lambda_C$ is the tuning parameter.  
The final CLIME estimator of $\Omega$ is defined by the symmetrization of $\widehat{\Omega}_1$:
\begin{align}
    &\widehat{\Omega}=(\widehat{\omega}_{ij}), \quad \text{where:}\\\notag
    &\widehat{\omega}_{ij}=\widehat{\omega}_{ji}=\widehat{\omega}^{1}_{ij}I\{|\widehat{\omega}^{1}_{ij}\leq |\widehat{\omega}^{1}_{ij}|\}+\widehat{\omega}^{1}_{ji}I\{|\widehat{\omega}^{1}_{ji}\leq |\widehat{\omega}^{1}_{ji}|\}.
\end{align}
The CLIME estimator has the property that the optimization problem defined in \eqref{eq:clime} can be decomposed into $p$ optimization problems defined as:
\begin{equation}\label{eq:ptimes}
    \min\|\boldsymbol{\beta}\|_1  \quad \text{subject to} \quad \|\Sigma\boldsymbol{\beta}-\boldsymbol{e}_i\|_{\infty}\leq \lambda_n,
\end{equation}
where $\boldsymbol{e}_i$ is defined as a standard unit vector in $\mathbb{R}^p$ with 1 in the $i$-th coordinate and 0 in the all other coordinates and $\boldsymbol{\beta}$ is a vector in $\mathbb{R}^p$. Solving \eqref{eq:ptimes} $p$ times, we obtain the CLIME estimator: $\{\widehat{\Omega}_1\}=\{\widehat{B}\}:=\{(\widehat{\boldsymbol{\beta}}_1,...,\widehat{\boldsymbol{\beta}}_p)\}$.

\section{Theoretical properties}\label{sec4}
We analyze the theoretical properties of the elliptical skew-(S)KEPTIC estimator. The results suggest that our estimator is reliable when dealing with highly skewed data. 
Our results show an exponential concentration rate in $\|\cdot\|_{\max}$. Moreover, we prove optimal bounds for shrinkage procedures. We show the results for the Dantzig selector, however, they can be easily extended to the CLIME and the Glasso procedures. 
\begin{theorem}\label{th_sk_rho}
Let $x_1,\dots,x_n$ be $n$ observations from $\boldsymbol{X}\sim MS\mathcal{E}_{p,p}(\Sigma^0,\Lambda,h^{(p)},0,\Gamma; f)$
restricted to the case $h^{(p)}=\phi^{(p)}(u)$, and let 
\[
\widehat{S}^{\rho,\alpha}_{ij} = 
2 \sin\!\left(\tfrac{\pi}{6}\widehat{\rho}_{ij}\right)\, \frac{1+\widehat{\alpha}_i \widehat{\alpha}_j}{\sqrt{(1+\widehat{\alpha}_i^2)(1+\widehat{\alpha}_j^2)}} \, \operatorname{I}(i \neq j)
\;+\; \operatorname{I}(i=j)
\]
denote the elliptical skew-SKEPTIC estimator. 
Then, with probability at least $1 - 2/p^2$, we have
\begin{equation}\label{eq:consistency_rho}
    \|\widehat{S}^{\rho,\alpha} - \Sigma^{0}\|_{\max}
    \leq (8\pi + 2 C_\alpha) \sqrt{\frac{\log p}{n}}.
\end{equation}
\end{theorem}

\begin{theorem}\label{th_sk_tau}
    Let $x_1,...,x_n$ be $n$ observations from $\boldsymbol{X}\sim MS\mathcal{E}_{p,p}(\Sigma^0,\Lambda,h^{(p)},0,\Gamma; f)$ and let 
    \[
\widehat{S}^{\tau,\alpha}_{ij} = 
\sin\!\left(\tfrac{\pi}{2}\widehat{\tau}_{ij}\right)\, \frac{1+\widehat{\alpha}_i \widehat{\alpha}_j}{\sqrt{(1+\widehat{\alpha}_i^2)(1+\widehat{\alpha}_j^2)}} \, \operatorname{I}(i \neq j)
\;+\; \operatorname{I}(i=j),
\] denote the elliptical skew-(S)KEPTIC estimator, then  with probability at least $1-2/p^2$ we have: 
    \begin{equation}\label{eq:consistency_tau}
        \|\widehat{S}^{\tau,\alpha}_{ij}-\Sigma^{0}_{ij}\|_{\max}\leq (2.45 \pi + 2 C_\alpha) \sqrt{\frac{\log p}{n}}.
    \end{equation}
\end{theorem}
Proofs of Theorems \ref{th_sk_rho} and \ref{th_sk_tau} are in \ref{appendix1}.
The elliptical skew-(S)KEPTIC estimator achieves convergence, thus it is possible to define the following result. 
\begin{theorem}\label{th: main}
When we substitute the estimated matrices $\srho$ and $\stau$ into the parametric Graphical Lasso (or the CLIME or graphical Dantzig selector), the estimators provide convergence, ensuring consistency and graph recovery. 
\end{theorem}
\begin{proof}
    The proof is based on the observation that the sample correlation matrix $\widehat{S}$ is a sufficient statistic for all three shrinkage procedures: the graphical lasso, graphical Dantzig selector, and CLIME. As \citet{yuan2010high} points out, a sufficient condition to ensure that an estimator is consistent is if there exists some constant $c$, such that $\mathbb{P}\left(\|\widehat{S}-\Sigma^{0}\|_{\max}>c\sqrt{\frac{\log p}{n}}\right)\leq 1-\frac{1}{p}$, which can be replaced by \eqref{eq:consistency_rho} and \eqref{eq:consistency_tau} in Theorems \ref{th_sk_rho} and \ref{th_sk_tau}. 
\end{proof}
The convergence of the elliptical skew-(S)KEPTIC estimator also means that the precision matrix is bounded. Similar results are found by \citet{cai2011constrained} and \citet{yuan2010high}. 
In particular, let $\widehat{\Omega}^{mse-s}$ denote the inverse correlation matrix using the elliptical skew-(S)KEPTIC estimator and consider the Dantzig selector as the shrinkage procedure to find $\widehat{\Omega}$. We introduce a class of inverse matrix defined as $\mathcal{M}_1(\kappa,\tau, M):=\{\Omega:\Omega\succ 0, \text{diag}(\Omega^{-1})=1,\|\Omega\|_{1}\leq \kappa,\frac{1}{\tau}\leq \lambda_{\min}(\Omega)\leq \lambda_{\max}(\Omega)\leq \tau, \deg(\Omega)\leq M\}$, where $\kappa,\tau>1$. Given Theorem \ref{th: main}, we obtain the following result: 
\begin{theorem}
    For $1\leq q < \infty$, there exists a constant $c$ that depends on $\kappa$, $\tau$, $\lambda_{\min}(\Omega)$ and $\lambda_{\max}(\Omega)$, such that
    \begin{equation}
        \sup_{\Omega\in \mathcal{M}_{1}(\kappa,\tau,M)}\|\widehat{\Omega}^{mse-s}-\Omega\|_{q}=O_P\left( M \sqrt{\frac{\log p}{n}}\right).
    \end{equation}
provided that $\lim _{n \rightarrow \infty} \frac{n}{M^2 \log d}=\infty$ and $\delta=c \sqrt{\frac{\log d}{n}}$, for sufficiently large $c$, where $\delta$ is the parameter in \eqref{theta}.
\end{theorem}
\begin{proof}
    This result follows \citet{yuan2010high}, which proves that the difference between the estimated precision matrix and the theoretical precision matrix is bounded by $\|\widehat{S}-\Sigma^{0}\|_{\max}$, where $\widehat{S}$ is the estimated correlation matrix and $\Sigma^{0}$ is the theoretical correlation matrix.
    In our case, to prove the bounds for the estimated correlation matrix, we just replace $\widehat{S}$ with $\srho$ and $\stau$ where the bounds are proved in Theorem \ref{th_sk_rho} and \ref{th_sk_tau}.
\end{proof}
\section{Experiments}\label{sec5}
In this section, we conduct an empirical analysis based on both synthetic and real data. The study of synthetic data is performed through numerical simulations, whereas the analysis of real data involves the estimation of undirected graphs for the stocks that compose the S\&P$500$ index.
\subsection{Numerical Simulations}
The implementation of the numerical simulations is carried out following these steps. Using the \texttt{R} package \textit{huge} (\citet{jiang2020package}), we generate the precision matrix $\Omega=(\Sigma^0)^{-1}$ using the "\textit{random}" structure, in which we randomly generate pairs of conditional dependencies in the off-diagonal elements, the sparsity of $\Omega$ is set to 0.02 to ensure the positive defitivness of the matrix so that the matrix can be inverted.  We consider the case of $p= 100$, and $n= 200$, and
we sample $n$ observations from $p$-variate normal distribution. We set the location parameter equal to zero, and $\Sigma^0=\Omega^{-1}$.

To evaluate the robustness of our estimators, we introduce a random contamination mechanism. For each marginal distribution, let $r \in(0,1)$ denote the contamination level. We randomly select $\lfloor n r\rfloor$ observations and replace them with values drawn from $\mathcal{Y} \sim N(0,5)$. 
Moreover, we apply the power transformation to the contaminated data, which is the following.  Let $g_0(t):=\operatorname{sign}(t)|t|^\gamma$ where $\gamma>0$ is a parameter. The power transformation for the $ j$th dimension is defined as
$$
g_j^0\left(z_j\right):=\frac{g_0\left(z_j-\mu_j\right)}{\sqrt{\int g_0^2\left(t-\mu_j\right) \phi\left(\left(t-\mu_j\right) / \sigma_j\right) d t}},
$$
where $\phi(\cdot)$ is the standard Gaussian density function.
These two procedures applied to the data induce a departure from normality, resulting in marginal distributions characterized by asymmetry and heavier tails compared to the Gaussian case. We consider three contamination levels (r = 0.05, 0.1, and 0.2), and three different parameters of the power function ($\gamma$ = 1.5, 2.5, and 3), considering a total of nine scenarios. 

In this analysis, we want to estimate the asymmetry of the data estimating the skewness of the data encoded in the matrix $D_\alpha$ of the distributions $f(\boldsymbol{X})\sim \operatorname{CSN}_{p,p}(0,\Sigma^0, D_\alpha L, 0, D_\kappa)$, and $f(\boldsymbol{X}) \sim \operatorname{SUT}_{p,p}(0,\Sigma^0, D_\alpha L, 0, D_\kappa)$. Regarding the closed skew-normal, we use the EM algorithm described in \cite{abe2021algorithm}, for the unified skew-$t$ we employ the likelihood-based method from the \texttt{R} package \textit{sn} developed by \citet{azzalini_pkg_sn}.

After the data generation, we compute the skew-(S)KEPTIC and the SKEPTIC estimators, which are then plugged into the glasso algorithm in \eqref{eq: glasso} to obtain $\Omega$. The glasso algorithm depends on a tuning parameter $\lambda_G$ whose choice is performed by applying the StARS method in \citet{liu2010stability},  which selects the regularization parameter based on the stability of the estimated graph under subsampling, and when the value is found for each dataset, we consider a set of 30 tuning parameters in the range $(\lambda_G-0.1, \lambda_G+0.1)$.
To evaluate the accuracy of the methods, we use the receiver operating curve ($\operatorname{ROC}$). We define the false positive rate and true positive rate for a fixed set of tuning parameters $\lambda_G$. When the tuning parameter is determined we estimate the graph: $\widehat{G}^{\lambda_G}=(V,\widehat{E}^{\lambda_G})$, to evaluate it, we define the number of false positives for a given $\lambda_G$, $\operatorname{FP}(\lambda_G)$, as the number of edges in $\widehat{E}^{\lambda_G}$ but not in $E$ and the number of false negatives, $\operatorname{FN}(\lambda_G)$, as the number of edges in $E$ but not in $\widehat{E}^{\lambda_G}$. We consider the rate of the number of false negatives and false positives:
\begin{equation}
    \operatorname{FNR}(\lambda_G):=\frac{\operatorname{FN}(\lambda_G)}{|E|}, \quad \text{and} \quad \operatorname{FPR}(\lambda_G):=\operatorname{FP}(\lambda_G) / \left[\left(\begin{array}{l}
d \\
2
\end{array}\right)-|E|\right].
\end{equation}
To demonstrate the comprehensive performance of the analyzed methods across the entire datasets, we define the ROC curve  
\begin{equation}\label{eq:roc}
(\operatorname{FNR}(\lambda_G), 1 - \operatorname{FPR}(\lambda_G)).
\end{equation}
Graphically, we show averaged ROC curves for all the considered scenarios over 50 trials, plotted using the coordinates in \eqref{eq:roc}. 
Fig. \ref{fig:numerical_simulations} illustrates that the skew-(S)KEPTIC estimators yield results comparable to those of the SKEPTIC estimator, which has been shown in \cite{liu2012high} to be optimal for a similar class of data. Importantly, incorporating skewness does not deteriorate the recovery of the true precision matrix $\Omega$. A more detailed comparison is provided in Table \ref{tab}, which reports the values (in percentage) of the Area Under the Curve (AUC), False Positive Rate (FPR), and False Negative Rate (FNR). Across all scenarios, taking into account both data contamination and power-function transformations, the elliptical skew-(S)KEPTIC estimators exhibit lower false positive rates and higher false negative rates, reflecting the smaller off-diagonal entries relative to those of the SKEPTIC estimator. Given that the tuning parameter $\lambda_G$ is fixed across all methods, it is not expected that the elliptical skew-(S)KEPTIC estimators outperform the SKEPTIC estimator in terms of overall graph recovery. Nevertheless, it is noteworthy that the skew-(S)KEPTIC estimators consistently exhibit superior performance compared to the Pearson correlation coefficient in the recovery of graphical structures.

These findings suggest that considering estimators that include the skewness of a skew-normal or a skew-t distribution and extend the {nonparanormal} SKEPTIC estimator, do not hinder the recovery of the $\Omega$ structure. On the contrary, they may offer a practical advantage due to the robustness across different data distributions. 
\begin{figure}[h!]
    \centering
    \includegraphics[width=0.9\textwidth, height=0.5\textheight, keepaspectratio]{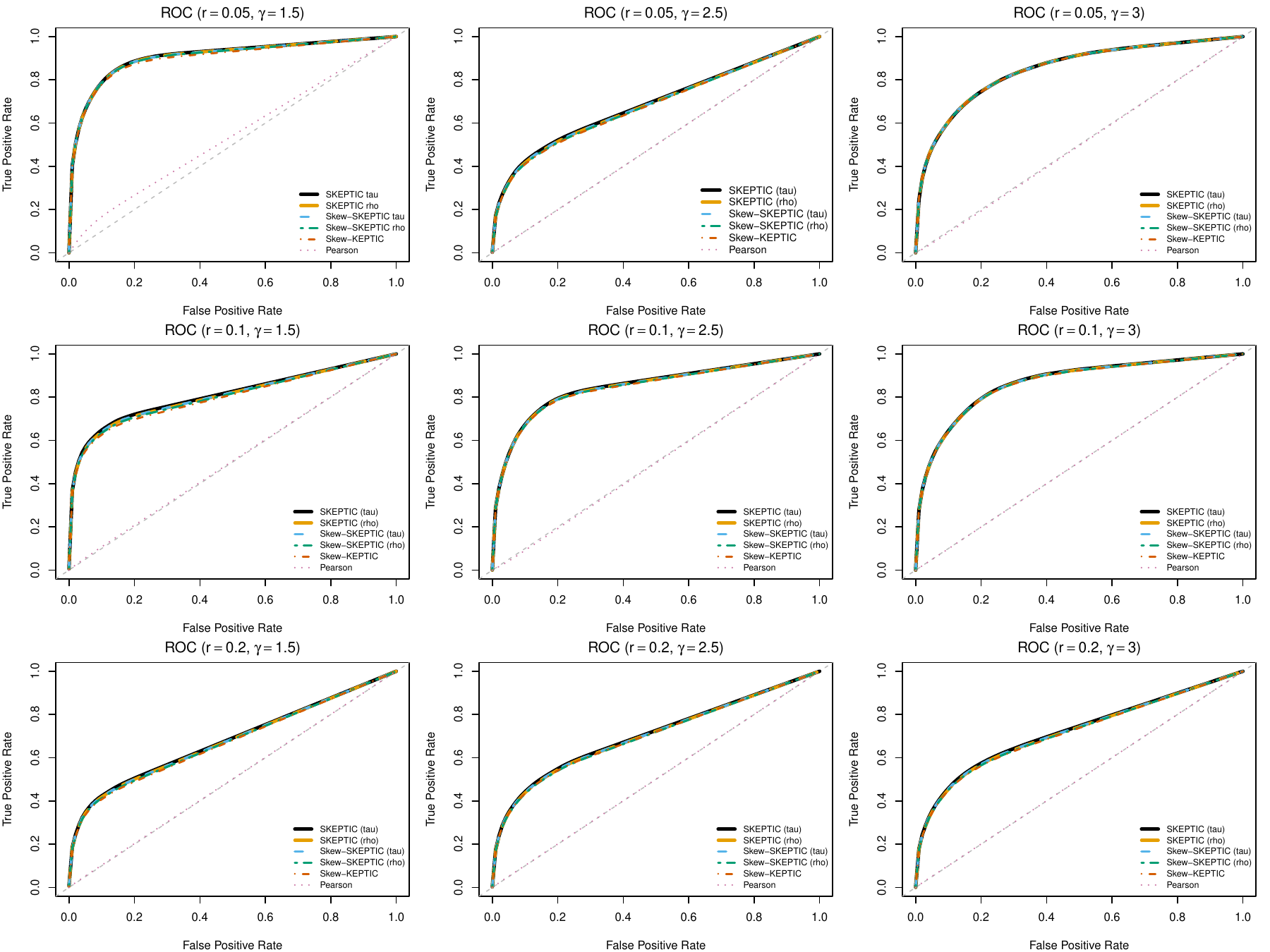}
    \caption{\textit{Comparison of average ROC curves, using the glasso algorithm, of the elliptical skew-(S)KEPTIC estimators with the SKEPTIC estimator. Results are determined over 50 trials. We consider, by row, three different data contamination and, by column, three different values of the parameter $\gamma$ that drives the power transformation.}}
    \label{fig:numerical_simulations}
\end{figure}

\begin{sidewaystable}[htbp]
    \centering
    \scriptsize 
    \setlength{\tabcolsep}{1.5pt} 
    \renewcommand{\arraystretch}{0.85} 
    \caption{\textit{Quantitative comparisons of the 6 correlation matrix estimators using different distributions to generate data, and compute the $\Omega$ matrix. The graphs are estimated using the glasso algorithm. The values are in percentages.}}
    \label{tab}
    \resizebox{\textwidth}{!}{
    \begin{tabular}{|c|cc|ccc|ccc|ccc|ccc|ccc|ccc|}
        \hline
        & & & \multicolumn{3}{|c|}{SKEPTIC ($\widehat{\rho}$)} & \multicolumn{3}{|c|}{SKEPTIC ($\widehat{\tau}$)} & \multicolumn{3}{|c|}{SKEW-SKEPTIC ($\widehat{\rho}$)} & \multicolumn{3}{|c|}{SKEW-SKEPTIC ($\widehat{\tau}$)} & \multicolumn{3}{|c|}{SKEW-KEPTIC} & \multicolumn{3}{|c|}{PEARSON} \\
        \hline
        & &  & AUC & FPR & FNR & AUC & FPR & FNR & AUC & FPR & FNR & AUC & FPR & FNR & AUC & FPR & FNR & AUC & FPR & FNR\\
        \hline
        \multirow[t]{2}{*}{} & r = 0.05 & $\gamma$ = 1.5 & 91.0 & 6.8 & 11.8 & 91.0 & 6.9 & 11.5 & 91.0 & 6.6 & 12.2 & 91.0 & 6.7 & 11.9 & 90.0 & 6.4 & 13.6 & 54.0 & 14.6 & 78.4 \\
        & & s.e & (1.3) & (2.1) & (3.5) & (1.4) & (2.1) & (3.5) & (1.4) & (2.0) & (3.6) & (1.4) & (2.0) & (3.6) & (1.7) & (2.0) & (4.4) & (1.4) & (11.7) & (9.9) \\
        \multirow[t]{2}{*}{} & r = 0.05 & $\gamma$ = 2.5 & 87.8 & 17.8 & 6.5 & 88.8 & 17.9 & 6.4 & 88.8 & 17.6 & 6.8 & 88.8 & 17.7 & 6.6 & 88.7 & 17.5 & 7.2 & 50.1 & 30.9 & 68.3 \\
        & & s.e & (2.7) & (6.0) & (3.2) & (2.8) & (6.0) & (3.2) & (2.7) & (5.9) & (3.2) & (2.7) & (6.0) & (3.4) & (2.6) & (6.0) & (3.3) & (0.9) & (15.4) & (14.3) \\
        \multirow[t]{2}{*}{} & r = 0.05 & $\gamma$ = 3 & 84.8 & 20.8 & 9.5 & 84.8 & 21.0 & 9.4 & 84.8 & 20.7 & 9.7 & 84.8 & 20.8 & 9.6 & 84.7 & 20.6 & 10.0 & 49.6 & 41.5 & 58.1  \\
        & & s.e & (12.3) & (10.5) & (14.3) & (12.3) & (10.5) & (14.4) & (12.3) & (10.6) & (14.3) & (12.3) & (10.6) & (14.3) & (12.3) & (10.6) & (14.2) & (0.7) & (13.4) & (12.8) \\
        \multirow[t]{2}{*}{} & r = 0.1 & $\gamma$ = 1.5 & 81.0 & 3.4 & 34.6 & 81.4 & 3.5 & 33.7 & 80.6 & 3.3 & 35.5 & 81.0 & 3.4 & 34.6 & 80.1 & 3.2 & 36.7 & 50.4 & 24.2 & 74.8  \\
        & & s.e & (3.0) & (2.0) & (6.9) & (3.0) & (2.0) & (6.9) & (3.0) & (1.9) & (6.8) & (3.0) & (1.9) & (7.0) & (3.0) & (1.9) & (6.8) & (0.5) & (17.9) & (17.2)  \\
        \multirow[t]{2}{*}{} & r = 0.1 & $\gamma$ = 2.5 & 84.8 & 7.8 & 22.8 & 85.1 & 7.7 & 22.2 & 84.6 & 7.4 & 23.4 & 84.8 & 7.5 & 22.9 & 84.5 & 7.3 & 23.8 & 49.5 & 39.8 & 59.9 \\
        & & s.e & (1.8) & (1.8) & (4.9) & (1.7) & (1.8) & (4.6) & (1.8) & (1.8) & (5.0) & (1.8) & (1.8) & (4.9) & (1.8) & (1.8) & (4.9) & (0.6) & (15.6) & (15.1) \\
        \multirow[t]{2}{*}{} & r = 0.1 & $\gamma$ = 3 & 86.6 & 12.8 & 14.0 & 86.7 & 13.0 & 13.6 & 86.4 & 12.7 & 14.4 & 86.6 & 12.8 & 14.0 & 86.3 & 12.6 & 14.8 & 50.0 & 31.6 & 67.9 \\
        & & s.e & (2.2) & (3.4) & (6.0) & (2.2) & (3.4) & (5.7) & (2.3) & (3.6) & (6.1) & (2.2) & (3.4) & (5.8) & (2.3) & (3.4) & (6.0) & (0.7) & (15.0) & (14.5) \\
        \multirow[t]{2}{*}{} & r = 0.2 & $\gamma$ = 1.5 & 68.1 & 2.9 & 61.0 & 68.5 & 2.9 & 60.1 & 67.8 & 2.8 & 61.6 & 68.2 & 2.8 & 60.9 & 67.5 & 2.7 & 62.4 & 50.2 & 31.8 & 67.7 \\
        & & s.e & (4.2) & (1.5) & (9.5) & (4.3) & (1.5) & (9.6) & (4.2) & (1.5) & (9.5) & (4.3) & (1.5) & (9.6) & (4.3) & (1.4) & (9.6) & (0.4) & (12.3) & (11.8) \\
        \multirow[t]{2}{*}{} & r = 0.2 & $\gamma$ = 2.5 & 70.5 & 5.3 & 53.7 & 70.8 & 5.4 & 53.0 & 70.2 & 5.2 & 54.4 & 70.5 & 5.3 & 53.7 & 70.3 & 5.2 & 54.3 & 50.1 & 31.1 & 68.4 \\
        & & s.e & (3.1) & (1.8) & (7.8) & (3.3) & (1.8) & (8.0) & (3.2) & (1.8) & (7.9) & (3.3) & (1.8) & (8.1) & (3.3) & (1.8) & (8.1) & (0.6) & (15.6) & (14.9) \\
        \multirow[t]{2}{*}{} & r = 0.2 & $\gamma$ = 3 & 71.9 & 6.6 & 49.6 & 72.3 & 6.6 & 48.8 & 71.7 & 6.4 & 50.2 & 72.0 & 6.5 & 49.5 & 71.5 & 6.3 & 50.7 & 50.1 & 30.9 & 68.5 \\
        & & s.e & (4.0) & (3.3) & (9.9) & (3.8) & (3.4) & (9.7) & (3.9) & (3.3) & (9.8) & (3.8) & (3.2) & (9.7) & (3.9) & (3.2) & (9.5) & (0.6) & (14.5) & (14.0)\\
        \hline
    \end{tabular}
    }
    \vspace{0.2cm}
    \textit{AUC = Area Under Curve, FPR = False Positive Rate, and FNR = False Negative Rate.}
\end{sidewaystable}

\newpage
\subsection{Empirical Analysis: S\&P500 Log-Returns}\label{sec6}
In the empirical analysis, we study the conditional dependence relations on daily returns of S\&P$500$ companies in the period from $04/01/2016$ to $31/12/2022$ downloaded from \textit{Yahoo! Finance}, using the proposed elliptical skew-(S)KEPTIC estimators and comparing them with the SKEPTIC estimator. After cleaning the data, the dataset consists of 454 stocks, each with 1762 observations. The log-returns of the stocks in the S\&P$500$ index are empirically skewed and exhibit heavier tails than the normal distribution (\citet{cont2001empirical}). 
\begin{figure}
    \centering
\includegraphics[width=\textwidth,height=7cm]{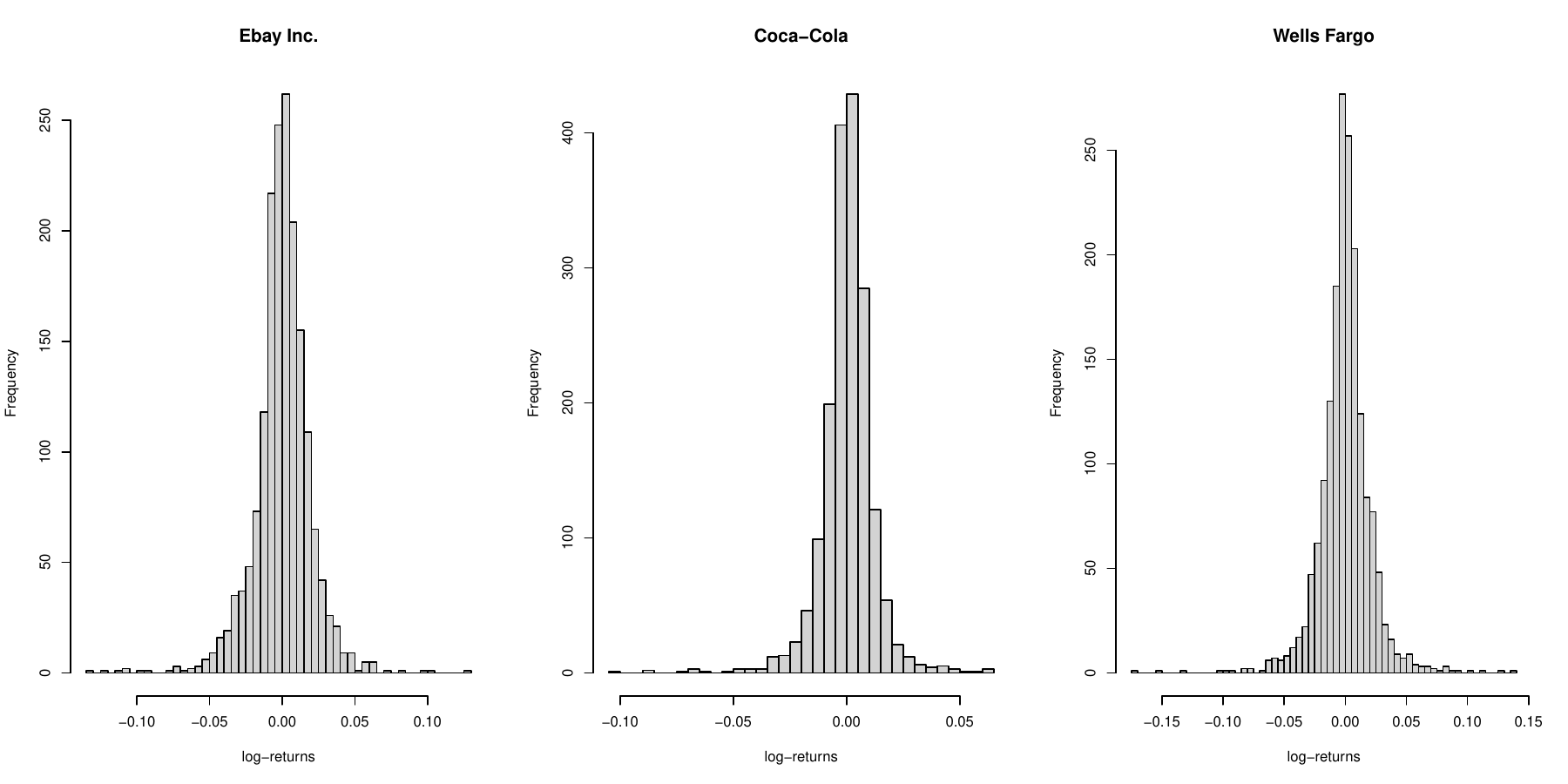}
    \caption{Illustration of leptokurtic distributions of daily log-returns, with a relevant skewness component. The stock data are from $04/01/2016$ to $31/12/2022$.}
    \label{fig: hists}
\end{figure}
As a motivating example, we benchmark the analysis of \citet{han2014scale}, reporting the asymmetry of log-returns in Figure \ref{fig: hists}, where we plot histograms of the returns of three stocks in our sample "eBay Inc.", "Coca-Cola" and "Wells Fargo"  belonging to three different sectors, respectively,  Consumer Discretionary, Consumer Staples, and Financials. We appreciate the departure from normality reflected in a leptokurtic distribution with a relevant skewness component. Moreover, we investigate the Gaussian assumption of the asset log-returns by computing normality tests with different significance levels at 1\% and 5\%. 
Table \ref{tab: tests} indicates a clear and complete deviation from normality in the stock data. This could be attributed to the presence of exceptional events, thus a huge presence of outlier values during the analyzed period, which lead to a departure from the assumption of normal distribution for log-returns. 

\begin{table}\caption{Normality tests S\&P500 log-returns data. This table illustrates the number of stocks that reject the null hypothesis of normality at 1\% and 5\% levels. }\label{tab: tests}
\centering
\begin{tabular}{cccc}
Significance Level & Kolmogorov-Smirnov & Shapiro-Wilk & Lilliefors \\
\toprule
1\% & 454 & 454 & 454 \\
5\% & 454 & 454 & 454 \\
\bottomrule
\end{tabular}
\end{table}

To estimate the graphs, given the large number of variables relative to observations, we determine the tuning parameter $\lambda_G$ using the StARS approach of \cite{liu2010stability}.
The tuning parameter $\lambda_G$ is set equal to $0.58$. 
To compute the elliptical skew-(S)KEPTIC estimators, we need to estimate the skewness parameters for the closed skew-normal distribution and for the unified skew-$t$ distribution. The skewness parameters are estimated using the same procedure as described in the numerical simulations.


Figure \ref{fig: undirected graph} shows the estimated graphs using SKEPTIC, elliptical skew-SKEPTIC, and elliptical skew-KEPITC estimators; all the correlation matrix estimators are based on Kendall's tau statistic.  
\begin{figure}[h!]
    \centering
\includegraphics[width=1\textwidth,height=16.5cm]{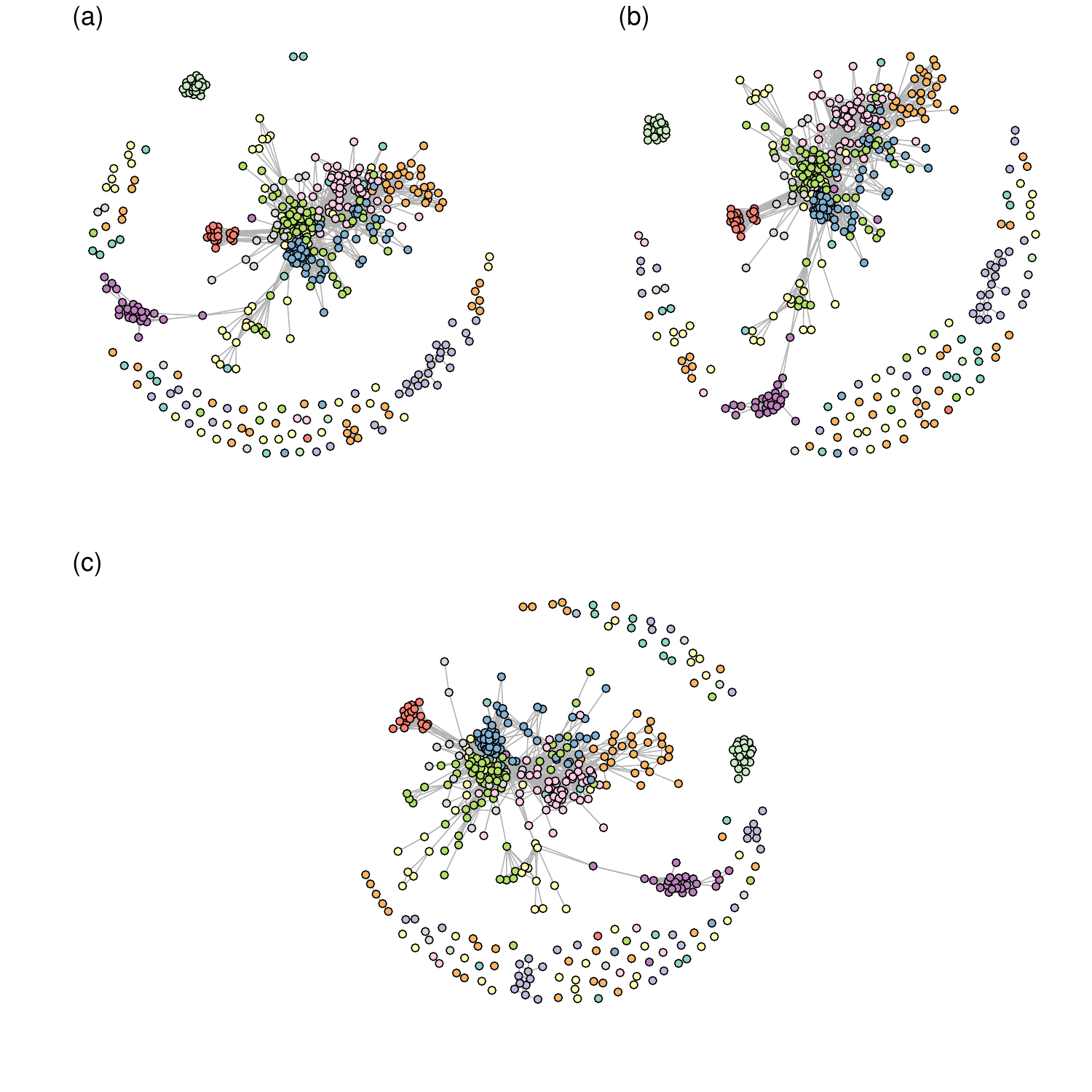}
    \caption{\textit{The resulting graphs of the log-returns of the S\&P 500 stocks from $04/01/2016$ to $31/12/2022$. Kendall's tau statistic is employed to estimate correlation. In the graph: (a) SKEPTIC estimator, (b) elliptical skew-SKEPTIC estimator, and (c) elliptical skew-KEPTIC estimator.}}
    \label{fig: undirected graph}
\end{figure}
The 454 stocks are categorized into 11 Global Industry Classification Standard (GICS) sectors, including \texttt{Consumer Discretionary} (51 stocks), \texttt{Consumer Staples} (33 stocks), \texttt{Energy} (21 stocks), \texttt{Financials} (67 stocks), \texttt{Health Care} (56 stocks), \texttt{Industrials} (69 stocks), \texttt{Information Technology} (59 stocks) \texttt{Telecommunications Services} (17 stocks), \texttt{Materials} (25 stocks), \texttt{Real Estate} (28 stocks), and \texttt{Utilities} (28 stocks).  

The stocks that belong to the same sector tend to cluster, and those at the periphery of the network are uncorrelated with the others during the considered period, thus, they are uncorrelated with the movement of the market. \citet{peralta2016network} show that those stocks have a greater weight in constructing a portfolio using the Markowitz approach. Therefore, in large-dimensions, a network approach can be valuable for identifying stocks with the least correlation to the market. That is, incorporating skewness and kurtosis components leads to a sparser network estimation compared to the normal case as reported in Table \ref{tab:sum_stat_net}. 
\begin{table}[!ht]
\centering
\caption{Summary statistics of the estimated networks using Kendall's tau statistic for the correlation matrix estimation. SK $=$ SKEPTIC, SNK $=$ elliptical skew-SKEPTIC, STK $=$ elliptical skew-KEPTIC. DIM = dimension of the dataset. The notation "$>$" indicates that one method has a greater number of edges compared to the other.}\label{tab:sum_stat_net}
\label{tab:table}
\vspace{0.1cm}
\begin{small}
\begin{tabular}{|l |c|c|c|c|c|c|c|}
\hline
\bfseries{Network} & & \multicolumn{3}{c|}{\bfseries \ Edges No.} & \multicolumn{3}{c|}{\bfseries Edges diff.} \\
\cline{2-8}
\hline
& DIM & SK & STK & SNK & SK$>$SNK & SK$>$STK & SNK$>$STK \\
\hline
S\&P 500 & 454  &  344 &  329 & 341 & 3    & 15 & 12 \\
\hline
\end{tabular}
\end{small}
\end{table}
\section{Concluding remarks}\label{sec7}

We propose rank-based correlation matrix estimators, named the elliptical skew-(S)KEPTIC estimators,  to provide semiparametric estimation of graphical models built on meta skew-elliptical distributions. 
The proposed framework includes two classes of estimators: the elliptical skew-KEPTIC estimator, which employs Kendall's tau statistic, and the elliptical skew-SKEPTIC estimator, which also uses Spearman's rho statistic to estimate the correlation matrix.

We study the theoretical properties of the rank-based correlation matrix estimators and establish concentration bounds. These bounds yield statistical rates of convergence under the elementwise maximum norm for both the correlation matrix estimation and the resulting estimation of the precision matrix. 

The empirical analysis and numerical simulations demonstrate that the elliptical skew-(S)KEPTIC estimators yield the optimal performance when the data deviates from normality.  The application to the log-returns of the S\&P500 constituent stocks shows resulting graphs that are sparser compared to those estimated under normality. 

Our findings demonstrate a reliable framework that is useful when dealing with data that exhibit high skewness which cannot be neglected, and its effectiveness is validated through the optimal graph recoveries of the numerical simulations.



\clearpage

\appendix
\section{Proofs of elliptical skew-(S)KEPTIC estimators}\label{appendix1}
\subsection{Proof of \textit{Theorem \ref{main th}}}
\begin{proof}
We consider $\boldsymbol{X}\sim MS\mathcal{E}_{p,p}\!\big(\Omega^{-1},\,D_\alpha L,\,h^{(p)},\,0,\,D_\kappa;\,f\big)$, where
$f(\boldsymbol{X}) \stackrel{d}{=} V_0^{-1/2} Z_0$ and
$Z_0 \sim \operatorname{CSN}_{p,p}\!\big(0,\,\Omega^{-1},\,D_{\alpha}L,\,0,\,D_{\kappa}\big)$, with $D_{\kappa}=I_p$ and $V_0=1$.
Under this framework, the stochastic representation of $\boldsymbol{X}$ coincides with~\eqref{eq: stoc_rep}; explicitly:
\begin{align*}
\boldsymbol{X}
&\stackrel{d}{=} L^{-1} D_\kappa^{-1/2} D_\alpha (I+D_\alpha^2)^{-1/2}\, \boldsymbol{U}
\;+\; L^{-1} D_\kappa^{-1/2} (I+D_\alpha^2)^{-1/2}\, \boldsymbol{V}.
\end{align*}
\[
{\boldsymbol X}
= D_\alpha (I+D_\alpha^2)^{-1/2}\, \boldsymbol{U}
\;+\; (I+D_\alpha^2)^{-1/2}\, \boldsymbol{V}.
\]
Here $\boldsymbol{V} \sim \operatorname{NPN}(0,\Sigma^0,f)$, where $\Sigma^0 = L^{-1}D_\kappa L$, meaning that its joint CDF is:
\[
F_{V}(v_1,\ldots,v_p)
= \Phi_{\Sigma^0}\!\big(\,\Phi^{-1}(F_{V_1}(v_1)),\,\ldots,\,\Phi^{-1}(F_{V_p}(v_p))\,\big),
\]
where each $F_{V_j}$ is estimated from normalized ranks as described in Section~\ref{sec: npn}.
Regarding the $U$ distribution:
\[
U_{ij}=F_{\mathrm{HN}}^{-1}(\widehat r_{ij})
=\Phi^{-1}\!\Big(\frac{\widehat r_{ij}+1}{2}\Big),
\]
where $F_{\mathrm{HN}}^{-1}$ is the standard half-normal quantile function.
By copula invariance under strictly increasing marginal transforms, $\operatorname{Copula}(U)=\operatorname{Copula}(V)$, (see e,g, \citealp{embrechts2002correlation}), this means that \ $U$ and $V$ share the same Gaussian copula dependence structure $\Sigma^0$.
Therefore, the joint CDF of $U$ (Gaussian copula with half-normal marginals) is:
\[
F_U(u_1,\ldots,u_p)
= \Phi_{\Sigma^0}\!\big(\,\Phi^{-1}(F_{\mathrm{HN}}(u_1)),\,\ldots,\,\Phi^{-1}(F_{\mathrm{HN}}(u_p))\,\big),\qquad u_i\ge 0,
\]
and $\Sigma^0$ can be estimated with the SKEPTIC estimator, where the generic element of $\Sigma^0_{ij}=\sin\Big(\frac{\pi}{2}\tau_{ij}\Big)$.
When $V_0\sim \mathrm{Gamma}(\nu/2,\nu/2)$, $\boldsymbol{X}$ reduces to a unified skew-$t$ (SUT) model and $V$ is a semiparametric $t$-copula. \cite{demarta2005t} show that the identity
$\Sigma^0_{ij}=\sin\!\big(\tfrac{\pi}{2}\tau_{ij}\big)$ still holds for $t$-copula, so the dependence structure is the same for the $t$ and half-$t$ distribution.

Let $A_0 = D_\alpha (I+D_\alpha^2)^{-1/2}$ and $B_0 = (I+D_\alpha^2)^{-1/2}$, and denote by
$S_\tau:=\sin\!\big(\tfrac{\pi}{2}\,T\big)$, where $T$ is the correlation matrix estimated with Kendall's tau statistic, the elliptical skew-KEPTIC estimator is:
\begin{align*}
S^{\tau,\alpha} \;\equiv\; \operatorname{Cor}({\boldsymbol X})
&= A_0\, S_\tau\, A_0 \;+\; B_0\, S_\tau\, B_0 \\
&= (I+D_\alpha^2)^{-1/2}\,\Big(S_\tau + D_\alpha S_\tau D_\alpha\Big)\,(I+D_\alpha^2)^{-1/2}.
\end{align*}
In the element-wise form, for $i\neq j$, the estimator is:
\[
S^{\tau,\alpha}_{ij}
= \sin\!\Big(\tfrac{\pi}{2}\tau_{ij}\Big)\;
\frac{1+\alpha_i \alpha_j}{\sqrt{(1+\alpha_i^2)(1+\alpha_j^2)}},
\quad \text{and:} \quad
S^{\tau,\alpha}_{ii}=1.
\]
\end{proof}
\subsection{Proof of Lemma \ref{Lemma_sp_rho}}
\begin{proof}
    We recall Lemma \ref{Lemma: SKEPTIC}, where for the normal distribution, and for $i \neq j$, we have: ${S}_{ij}=2\sin\left(\frac{\pi}{6}{\rho}_{ij}\right)=\sin\left(\frac{\pi}{2}{\tau}_{ij}\right)$. 
  Let $A_0 = D_\alpha (I+D_\alpha^2)^{-1/2}$ and $B_0 = (I+D_\alpha^2)^{-1/2}$, and denote by
$S_\rho:=2\sin\!\big(\tfrac{\pi}{6}\,P\big)$, where $P$ is the correlation matrix estimated using Spearman's rho statistic. Following the  proof of Theorem  \ref{main th}, the elliptical skew-SKEPTIC estimator is: 
\begin{align*}
S^{\rho,\alpha} \;\equiv\; \operatorname{Cor}({\boldsymbol X})
&= A_0\, S_\rho\, A_0 \;+\; B_0\, S_\rho\, B_0 \\
&= (I+D_\alpha^2)^{-1/2}\,\Big(S_\rho + D_\alpha S_\rho D_\alpha\Big)\,(I+D_\alpha^2)^{-1/2}.
\end{align*}
In the element-wise form, for $i\neq j$, the estimator is:
\[
S^{\rho,\alpha}_{ij}
= 2\sin\!\Big(\tfrac{\pi}{6}\rho_{ij}\Big)\;
\frac{1+\alpha_i \alpha_j}{\sqrt{(1+\alpha_i^2)(1+\alpha_j^2)}},
\qquad
S^{\rho,\alpha}_{ii}=1.
\]
\end{proof}

\subsection{Proof of \textit{Theorem \ref{th_sk_rho}}}
Let $B_{ij}=\dfrac{1+\alpha_i\alpha_j}{\sqrt{(1+\alpha_i^2)(1+\alpha_j^2)}}$ and,
for $i\neq j$, we write the target as
$S^{\rho,\alpha}_{ij}=2\sin\!\big(\tfrac{\pi}{6}\rho_{ij}\big)\,B_{ij}$,
and $S^{\rho,\alpha}_{ii}=1$.
For $i\neq j$:
\begin{align*}
\left|\widehat{S}^{\rho,\alpha}_{ij}-{S}^{\rho,\alpha}_{ij}\right| &\leq \left|\widehat{B}_{ij}2\sin\!\big(\tfrac{\pi}{6}\widehat{\rho}_{ij}\big)-B_{ij}2\sin\!\big(\tfrac{\pi}{6}\rho_{ij}\big)\right|\\
&\leq |\widehat{B}_{ij}|\left|2\sin\!\big(\tfrac{\pi}{6}\widehat{\rho}_{ij}\big)-2\sin\!\big(\tfrac{\pi}{6}\rho_{ij}\big)\right|+|\widehat{B}_{ij}-B_{ij}|,
\end{align*}
since, by construction, $|\widehat{B}_{ij}|\leq 1$. 
From \cite{liu2012high}, we know that, with probability $1-1/p^2$, we have: 
$$
\max_{ij}\left|2\sin\!\big(\tfrac{\pi}{6}\widehat{\rho}_{ij}\big)-2\sin\!\big(\tfrac{\pi}{6}\rho_{ij}\big)\right|\leq 8 \pi \sqrt{\frac{\log p}{n}} .
$$
Regarding the term $|\widehat{B}_{ij}-B_{ij}|$, we define $f(a,b)=\dfrac{1+ab}{\sqrt{(1+a^2)(1+b^2)}}$ so that
$\widehat B_{ij}=f(\widehat\alpha_i,\widehat\alpha_j)$ and $B_{ij}=f(\alpha_i,\alpha_j)$
.
A direct calculation shows that:
\[
\frac{\partial f}{\partial a}(a,b)=\frac{b-a}{(1+a^2)^{3/2}(1+b^2)^{1/2}},
\qquad
\frac{\partial f}{\partial b}(a,b)=\frac{a-b}{(1+a^2)^{1/2}(1+b^2)^{3/2}},
\]
hence $|\partial f/\partial a|\le 1$ and $|\partial f/\partial b|\le 1$ on $\mathbb R^2$.
By the mean value theorem in two variables,
\begin{equation}\label{eq:B-Lipschitz}
|\widehat B_{ij}-B_{ij}|
=|f(\widehat\alpha_i,\widehat\alpha_j)-f(\alpha_i,\alpha_j)|
\le |\widehat\alpha_i-\alpha_i|+|\widehat\alpha_j-\alpha_j|.
\end{equation}
Therefore, we have that: 
$$
\max_{ij}\left|\widehat{S}^{\rho,\alpha}_{ij}-{S}^{\rho,\alpha}_{ij}\right|\leq 8\pi\sqrt{\frac{\log p}{n}}+ 2\max_{1\leq l \leq p}|\widehat{\alpha}_l - \alpha_l|
$$
Define $A:=\max_{1\le \ell\le p}|\widehat{\alpha}_\ell-\alpha_\ell|$ and fix an arbitrary threshold $t_\alpha>0$.
On the event $\{A\le t_\alpha\}$ we have $2A\le 2t_\alpha$. Hence:
\[
\mathbb{P}\big(2A\le 2t_\alpha\big)
\;\ge\; 1-\mathbb{P}(A>t_\alpha).
\]
Equivalently, with probability at least $1-\mathbb{P}(A>t_\alpha)$, the plug\mbox{-}in error term satisfies
\[
2\,\max_{1\le \ell\le p}\big|\widehat{\alpha}_\ell-\alpha_\ell\big|\;\le\;2t_\alpha.
\]

Exploiting Boole's inequality, we obtain:
\[
\mathbb P\!\left(\|\widehat S^{\rho,\alpha}-\Sigma^0\|_{\max}
\;\le\; 8\,\pi\,\sqrt{\frac{\log p}{n}} \;+\; 2t_\alpha\right)
\;\ge\; 1-\frac{1}{p}-\mathbb P(A>t_\alpha).
\]
In particular, if $t_\alpha$ is chosen so that $\mathbb P(A>t_\alpha)\le 1/p^2$,
then with probability at least $1-2/p^2$,
\[
\|\widehat S^{\rho,\alpha}-\Sigma^0\|_{\max}
\;\le\;8\,\pi\,\sqrt{\frac{\log p}{n}} \;+\; 2t_\alpha.
\]
Moreover, if we assume that the $\alpha$ parameters are correctly estimated using likelihood-based methods, we have that $A=O_{P}\!\left(\sqrt{\frac{\log p}{n}}\right)$, 
taking $t_\alpha=C_\alpha\sqrt{\frac{\log p}{n}}$, where $C_\alpha$ is a constant value that depends from $\alpha$, we obtain:
\[
\|\widehat S^{\rho,\alpha}-\Sigma^0\|_{\max}
\;\le\; \big(8\,\pi + 2C_\alpha\big)\sqrt{\frac{\log p}{n}}
\quad,
\]
with probability greater than $1-\tfrac{2}{p^2}$.
\subsection{Proof of \textit{Theorem \ref{th_sk_tau}}}
From \cite{liu2012high}, we know that with probability $1-1/p^2$ we have that:
$$
\left|\widehat{S}_{ij}^\tau-\Sigma_{ij}^0\right| \leq 2.45 \pi \sqrt{\frac{\log p}{n}} 
$$
The results for our estimator become:
$$
\max_{ij}\left|\widehat{S}^{\tau,\alpha}_{ij}-{S}^{\tau,\alpha}_{ij}\right|\leq 2.45\pi\sqrt{\frac{\log p}{n}}+ 2\max_{1\leq l \leq p}|\widehat{\alpha}_l - \alpha_l|,
$$
where choosing a probability $1-1/p^2$, and assuming the convergence of the likelihood-based estimator, we obtain:
\[
\|\widehat S^{\tau,\alpha}-\Sigma^0\|_{\max}
\;\le\; \big(2.45\,\pi + 2C_\alpha\big)\sqrt{\frac{\log p}{n}},
\]
with probability greater than $1-\tfrac{2}{p^2}$.

\clearpage
\bibliographystyle{apalike}  
\bibliography{jmva.bib} 

@article{liu2009nonparanormal,
  title={The nonparanormal: Semiparametric estimation of high dimensional undirected graphs},
  author={Liu, Han and Lafferty, John and Wasserman, Larry},
  journal={Journal of Machine Learning Research},
  volume={10},
  pages={2295--2328},
  year={2009},
  publisher={Microtome Publishing}
}

@article{liu2012high,
  title={High Dimensional Semiparametric Gaussian Copula Graphical Models},
  author={Liu, Han and Han, Fang and Yuan, Ming and Lafferty, John and Wasserman, Larry},
  journal={The Annals of Statistics},
  volume={40},
  number={4},
  pages={2293--2326},
  year={2012}
}

@article{zareifard2016skew,
  title={A skew Gaussian decomposable graphical model},
  author={Zareifard, Hamid and Rue, H{\aa}vard and Khaledi, Majid Jafari and Lindgren, Finn},
  journal={Journal of Multivariate Analysis},
  volume={145},
  pages={58--72},
  year={2016},
  publisher={Elsevier}
}

@article{kruskal1958ordinal,
  title={Ordinal measures of association},
  author={Kruskal, William H},
  journal={Journal of the American Statistical Association},
  volume={53},
  number={284},
  pages={814--861},
  year={1958},
  publisher={Taylor \& Francis}
}

@article{cont2001empirical,
  title={Empirical properties of asset returns: stylized facts and statistical issues},
  author={Cont, Rama},
  journal={Quantitative Finance},
  volume={1},
  number={2},
  pages={223–236},
  year={2001},
  publisher={IOP Publishing}
}

@article{friedman2008sparse,
  title={Sparse inverse covariance estimation with the graphical lasso},
  author={Friedman, Jerome and Hastie, Trevor and Tibshirani, Robert},
  journal={Biostatistics},
  volume={9},
  number={3},
  pages={432--441},
  year={2008},
  publisher={Oxford University Press}
}

@article{cai2011constrained,
  title={A constrained $\ell_1$ minimization approach to sparse precision matrix estimation},
  author={Cai, Tony and Liu, Weidong and Luo, Xi},
  journal={Journal of the American Statistical Association},
  volume={106},
  number={494},
  pages={594--607},
  year={2011},
  publisher={Taylor \& Francis}
}

@article{dominguez2003multivariate,
  title={The multivariate closed skew normal distribution},
  author={Domìnguez-Molina, J and Gonzàlez-Farìas, G and Gupta, AK},
  journal={Technical Report 03-12},
  year={2003},
  publisher={Department of Mathematics and Statistics, Bowling Green State University}
}

@article{banerjee2008model,
  title={Model selection through sparse maximum likelihood estimation for multivariate Gaussian or binary data},
  author={Banerjee, Onureena and El Ghaoui, Laurent and d'Aspremont, Alexandre},
  journal={The Journal of Machine Learning Research},
  volume={9},
  pages={485--516},
  year={2008},
  publisher={JMLR. org}
}

@article{arellano2010skew-t,
  title={Multivariate extended skew-t distributions and related families},
  author={Arellano-Valle, Reinaldo B and Genton, Marc G},
  journal={Metron},
  volume={68},
  pages={201--234},
  year={2010},
  publisher={Springer}
}

@article{liu2012transelliptical,
  title={Transelliptical graphical models},
  author={Liu, Han and Han, Fang and Zhang, Cun-hui},
  journal={Advances in Neural Information Processing Systems},
  volume={25},
  year={2012}
}

@article{meinshausen2006high,
  title={High-dimensional graphs and variable selection with the Lasso},
  author={Meinshausen, Nicolai and B{\"u}hlmann, Peter},
  journal={The Annals of Statistics},
  volume={34},
  number={1},
  pages={1436--1462},
  year={2006}
}

@article{yuan2010high,
  title={High dimensional inverse covariance matrix estimation via linear programming},
  author={Yuan, Ming},
  journal={The Journal of Machine Learning Research},
  volume={11},
  pages={2261--2286},
  year={2010},
  publisher={JMLR. org}
}

@article{arellano2010skew-e,
  title={Multivariate unified skew-elliptical distributions},
  author={Arellano-Valle, Reinaldo B and Genton, Marc G},
  journal={Chilean Journal of Statistics},
  volume={1},
  number={1},
  pages={17--33},
  year={2010},
  publisher={Citeseer}
}

@Manual{jiang2020package,
    title = {huge: High-Dimensional Undirected Graph Estimation},
    author = {Haoming Jiang and Xinyu Fei and Han Liu and Kathryn Roeder and John Lafferty and Larry Wasserman and Xingguo Li and Tuo Zhao},
    year = {2021},
    note = {R package version 1.3.5},
    url = {https://CRAN.R-project.org/package=huge},
  }

@Manual{azzalini_pkg_sn,
    title = {The {R} package \texttt{sn}: The skew-normal and related distributions such as the skew-$t$ and the {SUN} (version 2.1.1).},
    author = {Azzalini A. Azzalini},
    address = {Universit\`a degli Studi di Padova, Italia},
    year = {2023},
    note = {Home page: \url{http://azzalini.stat.unipd.it/SN/}},
    URL = {https://cran.r-project.org/package=sn},
  }

@article{han2014scale,
  title={Scale-invariant sparse \textit{PCA} on high-dimensional meta-elliptical data},
  author={Han, Fang and Liu, Han},
  journal={Journal of the American Statistical Association},
  volume={109},
  number={505},
  pages={275--287},
  year={2014},
  publisher={Taylor \& Francis}
}

@article{nghiem2022estimation,
  title={Estimation of graphical models for skew continuous data},
  author={Nghiem, Linh H and Hui, Francis KC and M{\"u}ller, Samuel and Welsh, Alan H},
  journal={Scandinavian Journal of Statistics},
  volume={49},
  number={4},
  pages={1811--1841},
  year={2022},
  publisher={Wiley Online Library}
}

@article{sheng2023skewed,
  title={On skewed Gaussian graphical models},
  author={Sheng, Tianhong and Li, Bing and Solea, Eftychia},
  journal={Journal of Multivariate Analysis},
  volume={194},
  pages={105-129},
  year={2023},
  publisher={Elsevier}
}

@article{azzalini1985class,
  title={A class of distributions which includes the normal ones},
  author={Azzalini, Adelchi},
  journal={Scandinavian Journal of Statistics},
  volume={12},
  pages={171--178},
  year={1985}
}

@book{fang2018symmetric,
  title={Symmetric multivariate and related distributions},
  author={Fang, Kai Wang},
  year={2018},
  publisher={CRC Press}
}

@article{peralta2016network,
  title={A network approach to portfolio selection},
  author={Peralta, Gustavo and Zareei, Abalfazl},
  journal={Journal of Empirical Finance},
  volume={38},
  pages={157--180},
  year={2016},
  publisher={Elsevier}
}

@article{abe2021algorithm,
  title={$\operatorname{EM}$ algorithm using overparameterization for the multivariate skew-normal distribution},
  author={Abe, Toshihiro and Fujisawa, Hironori and Kawashima, Takayuki and Ley, Christophe},
  journal={Econometrics and Statistics},
  volume={19},
  pages={151--168},
  year={2021},
  publisher={Elsevier}
}

@article{hult2002multivariate,
  title={Multivariate extremes, aggregation and dependence in elliptical distributions},
  author={Hult, Henrik and Lindskog, Filip},
  journal={Advances in Applied Probability},
  volume={34},
  number={3},
  pages={587--608},
  year={2002},
  publisher={Cambridge University Press}
}

@article{liu2010stability,
  title={Stability approach to regularization selection (stars) for high dimensional graphical models},
  author={Liu, Han and Roeder, Kathryn and Wasserman, Larry},
  journal={Advances in Neural Information Processing Systems},
  volume={23},
  year={2010}
}

@article{rossell2021dependence,
  title={Dependence in elliptical partial correlation graphs},
  author={Rossell, David and Zwiernik, Piotr},
  journal={Electronic Journal of Statistics},
  volume={15},
  number={2},
  pages={4236--4263},
  year={2021},
  publisher={The Institute of Mathematical Statistics and the Bernoulli Society}
}

@article{embrechts2002correlation,
  title={Correlation and dependence in risk management: properties and pitfalls},
  author={Embrechts, Paul and McNeil, Alexander and Straumann, Daniel},
  journal={Risk Management: Value at Risk and Beyond},
  volume={1},
  pages={176--223},
  year={2002},
  publisher={New York}
}

@article{demarta2005t,
  title={The t copula and related copulas},
  author={Demarta, Stefano and McNeil, Alexander J},
  journal={International Statistical Review},
  volume={73},
  number={1},
  pages={111--129},
  year={2005},
  publisher={Wiley Online Library}
}
\end{document}